\documentclass[a4paper,12pt,twoside]{amsart} 

\usepackage{indentfirst}
\usepackage{amssymb}
\usepackage{amsmath}
\usepackage{amsfonts} 
\usepackage{latexsym}
\usepackage{pgf,tikz}
\usepackage{breqn}
\usepackage{upgreek}
\usepackage{graphicx}
\usepackage{epstopdf}
\usepackage{latexsym}
\usepackage[belowskip=0pt,aboveskip=0pt]{caption}
\usepackage[bottom]{footmisc}
\usepackage[bookmarks=false]{hyperref}
\usepackage[top=1.5in,bottom=1.2in,left=1.2in,right=1.2in]{geometry}

\usetikzlibrary{arrows}

\theoremstyle{plain}

\newtheorem{theorem}{Theorem}

\DeclareMathOperator{\vor}{Vor}

\allowdisplaybreaks

\begin{document}

\title{A note on rectangle covering with congruent disks}

\author{Emanuele Tron}
\address{Scuola Normale Superiore, piazza dei Cavalieri 7, 56126 Pisa, Italy}
\email{emanuele.tron@sns.it}

\subjclass[2010]{52C15, 05B40}

\begin{abstract}
In this note we prove that, if $S_n$ is the greatest area of a rectangle which can be covered with $n$ unit disks, then $2 \leq S_n/n <3 \sqrt 3/2$, and these are the best constants; moreover, for $\Delta(n):=(3 \sqrt 3/2) n-S_n$, we have $0.727384  < \liminf \Delta(n)/\sqrt n < 2.121321$ and $0.727384 < \limsup \Delta(n)/\sqrt n < 4.165064 $.
\end{abstract}

\maketitle

The problem of covering sets in the plane with figures of prescribed shape has been extensively studied in literature--even though the dual packing problem received comparatively much more attention--both from the theoretical and computational viewpoint, also in virtue of its practical applications. In this note we study the extreme values for the area of a rectangle covered by a fixed number of congruent disks. Our aim is here to give precise bounds for the maximum value of this area.

Let then $S_n$ be the greatest area of a rectangle which can be covered with $n$ closed disks of unit radius. Here we prove the following two facts.
\begin{theorem} \label{1} For every $n \in \mathbb N$, \[ 2n \leq S_n < \frac {3 \sqrt 3} 2 n.\] These are the best possible constants: $\min_{n \in \mathbb N} S_n/n=2$ and $\limsup_{n \rightarrow \infty}  S_n/n= 3 \sqrt 3 /2$.
\end{theorem}
Define moreover
\[ \Delta(n):= \frac{3 \sqrt 3}{2}n-S_n, \quad \alpha:= \liminf_{n \rightarrow \infty} \frac {\Delta(n)} {\sqrt n}, \quad  \beta:= \limsup_{n \rightarrow \infty} \frac {\Delta(n)} {\sqrt n}. \]
Then one has
\begin{theorem} \label{2} \[ 0.727384 \ldots  \leq \alpha  \leq 2.121320 \ldots \]
\[ 0.727384 \ldots\leq \beta \leq  4.165063 \ldots \]
\end{theorem}
First, let $\mathcal C_1$, $\dots$, $\mathcal C_n$ be the circles covering a rectangle (that we treat as fixed) and $O_1$, $\dots$, $O_n$ their centers, and recall that the \textit{Voronoi cell} $\vor _i$ of the circle $\mathcal C_i$ is the set of points $Q$ inside the rectangle such that the distance of $Q$ from $O_j$ is greater or equal than its distance from $O_i$ for all $j\neq i$.
\begin{proof}[Proof of Theorem \ref{1}.]
The leftmost inequality is trivial. Just take a rectangle built by juxtaposing $n$ squares, each inscribed in a circle as in Figure \ref{fig:3}: each square has area $2$, hence the rectangle has area $2n$. The constant $2$ is the best possible one because the largest rectangle with fixed circumcircle is the square, that is we have equality for $n=1$.

\begin{figure}[t]
\scalebox{0.6}{
\begin{tikzpicture}[line cap=round,line join=round,>=triangle 45,x=1.0cm,y=1.0cm]
\clip(17,17.5) rectangle (33,22.23);
\draw [line width=1pt] (20.32,19.94) circle (1.75cm);
\draw [line width=1pt] (19.08,18.7)-- (19.08,21.18);
\draw [line width=1pt] (21.56,18.7)-- (21.56,21.18);
\draw [line width=1pt] (22.8,19.94) circle (1.75cm);
\draw [line width=1pt] (25.28,19.94) circle (1.75cm);
\draw [line width=1pt] (27.76,19.94) circle (1.75cm);
\draw [line width=1pt] (30.24,19.94) circle (1.75cm);
\draw [line width=1pt] (19.08,18.7)-- (31.48,18.7);
\draw [line width=1pt] (19.08,21.18)-- (31.48,21.18);
\draw [line width=1pt] (31.48,21.18)-- (31.48,18.7);
\draw [line width=1pt] (29,18.7)-- (29,21.18);
\draw [line width=1pt] (26.52,21.18)-- (26.52,18.7);
\draw [line width=1pt] (24.04,18.7)-- (24.04,21.18);
\end{tikzpicture}
}\caption{Construction of a rectangle for $S_n \geq 2n$.} \label{fig:3}
\end{figure}
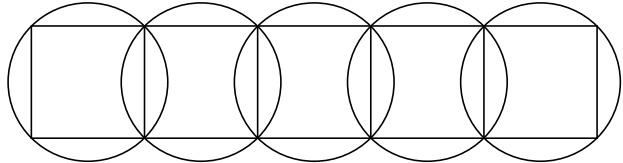

For the other inequality, we adapt the argument of \cite{K}.
Each disk of the covering has an attached Voronoi cell, contained in the circle, which may be assumed to be nonempty. Each Voronoi cell is a convex polygon whose sides are either parts of chords formed by the pairwise intersection of circles, or part of the sides of the rectangle. Each point inside the rectangle, except the boundaries of cells, is contained in exactly one Voronoi cell.

Having said that, we proceed with a modified version of a lemma from \cite{G}. Here we treat the covering as a planar graph whose faces are the Voronoi cells and whose edges and vertices are those of the cells. Then, under the assumption that every vertex of the net is contained in at least three sides, except for exactly four vertices which belong to two sides (which is the case for our covering, where the four exceptional vertices are those of the rectangle), the average number of sides of a cell is less than $6-2 \sqrt{ 2 /n}+ 2/ n$.

To see why, note first that, if $v$ and $e$ are respectively the numbers of vertices and edges in the net, Euler's formula reads $v-e+n=1$. Since every edge contains two vertices, double-counting the sides with the aid of the hypothesis gives $3(v-4)+8 \leq 2e$. Combining the two provides the inequality $e \leq 3n+1$. 

Let now $e_i$ be the number of sides of $\vor_i$: since some edges, but not all, belong to two faces, $\sum_{i=1}^n e_i <2e \leq 6n+2$. Moreover, we can obtain a lower bound on the number of sides which belong to one face only. Every edge which is part of the boundary of the rectangle has this property, and each of these edges has length at most $2$. The perimeter of the rectangle is at least $4 \sqrt{S_n} \geq 4 \sqrt{2n}$, hence there are at least $2 \sqrt{2n}$ of these edges. 

The average number of sides of the cells is then \[ \frac 1 n \sum_{i=1}^n e_i  \leq \frac 1 n \left ( 6n+2 -2 \sqrt{2n}\right )=6-\frac{2 \sqrt 2}{ \sqrt n}+\frac 2 n . \]

If we let $V_i$ be the number of cells in the covering which have exactly $i$ sides, so that $\sum_{i=3}^{\infty}V_i=n$, the previous inequality can be expressed as \[\sum_{i=3}^\infty iV_i <\sum_{i=3}^\infty (6-\epsilon(n))V_i,\quad\text{that is,}\quad \sum_{i=3}^\infty (6-i)V_i>\epsilon(n) \sum_{i=3}^\infty V_i=n \epsilon(n), \] where we have set for convenience $\epsilon(n):=2 \sqrt{2/ n}- 2/ n$. 

Let $K_i=(i/2)\sin (2 \uppi/i)$, which for integer $i$ is the area of a regular $i$-agon with unit circumradius. The function $x \mapsto K_x$ is strictly increasing and concave, so every line through $(i, K_i)$ and $(i+1, K_{i+1})$ lies above all other points of $(j, K_j)_{j \in \mathbb N}$. Taking $i=5$ gives us $K_j \leq (K_6-K_5)j-5K_6+6K_5$ for every $j \geq 3$.

By the inequality we just obtained,  \begin{gather*}\sum_{i=3}^{\infty} (K_6-K_i)V_i\geq \sum_{i=3}^\infty (K_6-(K_6-K_5)i+5K_6-6K_5)V_i\\=(K_6-K_5)\sum_{i=3}^\infty (6-i)V_i>(K_6-K_5)n \epsilon(n)=:R(n).\end{gather*}
Observe that the cyclic $d$-agon with fixed circumcircle and greatest area is the regular one, so that
\[\text{area of the rectangle} \leq \sum_{i=3}^\infty K_i V_i<K_6\sum_{i=3}^\infty V_i-R(n)=K_6 n-R(n).\]
This implies Theorem \ref{1} as long as $R(n)=2(K_6-K_5)(\sqrt{2n}-1)>0$, which is trivially true.

\begin{figure}[t]
\scalebox{0.7}{
\begin{tikzpicture}[line cap=round,line join=round,>=triangle 45,x=1.2701084170855053cm,y=1.2701084170855053cm]
\clip(4.26,13.15) rectangle (14.51,21.35);
\draw [line width=1pt] (5.11,13.96)-- (5.91,13.5);
\draw [line width=1pt] (5.91,13.5)-- (6.72,13.96);
\draw [line width=1pt] (6.72,13.96)-- (7.52,13.49);
\draw [line width=1pt] (7.52,13.49)-- (8.33,13.96);
\draw [line width=1pt] (8.33,13.96)-- (9.14,13.49);
\draw [line width=1pt] (9.14,13.49)-- (9.94,13.96);
\draw [line width=1pt] (9.94,13.96)-- (10.75,13.49);
\draw [line width=1pt] (10.75,13.49)-- (11.55,13.95);
\draw [line width=1pt] (11.55,13.95)-- (12.36,13.49);
\draw [line width=1pt] (12.36,13.49)-- (13.16,13.95);
\draw [line width=1pt] (13.16,13.95)-- (13.17,14.88);
\draw [line width=1pt] (13.17,14.88)-- (12.36,15.35);
\draw [line width=1pt] (12.36,15.35)-- (12.36,16.28);
\draw [line width=1pt] (12.36,16.28)-- (13.17,16.74);
\draw [line width=1pt] (13.17,16.74)-- (13.97,16.28);
\draw [line width=1pt] (13.97,15.35)-- (13.17,14.88);
\draw [line width=1pt] (13.97,15.35)-- (13.97,16.28);
\draw [line width=1pt] (12.36,15.35)-- (11.55,14.88);
\draw [line width=1pt] (11.55,14.88)-- (10.75,15.35);
\draw [line width=1pt] (10.75,15.35)-- (9.94,14.89);
\draw [line width=1pt] (9.94,14.89)-- (9.14,15.35);
\draw [line width=1pt] (9.14,15.35)-- (8.33,14.89);
\draw [line width=1pt] (8.33,14.89)-- (7.53,15.35);
\draw [line width=1pt] (7.53,15.35)-- (6.72,14.89);
\draw [line width=1pt] (6.72,14.89)-- (5.91,15.36);
\draw [line width=1pt] (5.91,15.36)-- (5.11,14.89);
\draw [line width=1pt] (5.11,14.89)-- (5.11,13.96);
\draw [line width=1pt] (6.72,14.89)-- (6.72,13.96);
\draw [line width=1pt] (8.33,14.89)-- (8.33,13.96);
\draw [line width=1pt] (9.94,14.89)-- (9.94,13.96);
\draw [line width=1pt] (11.55,13.95)-- (11.55,14.88);
\draw [line width=1pt] (10.75,15.35)-- (10.75,16.28);
\draw [line width=1pt] (10.75,16.28)-- (11.56,16.74);
\draw [line width=1pt] (11.56,16.74)-- (12.36,16.28);
\draw [line width=1pt] (10.75,16.28)-- (9.94,16.75);
\draw [line width=1pt] (9.94,16.75)-- (9.14,16.28);
\draw [line width=1pt] (9.14,16.28)-- (9.14,15.35);
\draw [line width=1pt] (9.14,16.28)-- (8.33,16.75);
\draw [line width=1pt] (8.33,16.75)-- (7.53,16.28);
\draw [line width=1pt] (7.53,16.28)-- (7.53,15.35);
\draw [line width=1pt] (7.53,16.28)-- (6.72,16.75);
\draw [line width=1pt] (6.72,16.75)-- (5.92,16.29);
\draw [line width=1pt] (5.92,16.29)-- (5.91,15.36);
\draw [line width=1pt] (5.92,16.29)-- (5.11,16.75);
\draw [line width=1pt] (5.11,16.75)-- (5.11,17.68);
\draw [line width=1pt] (5.11,17.68)-- (5.92,18.15);
\draw [line width=1pt] (5.92,18.15)-- (6.72,17.68);
\draw [line width=1pt] (6.72,17.68)-- (6.72,16.75);
\draw [line width=1pt] (6.72,17.68)-- (7.53,18.15);
\draw [line width=1pt] (7.53,18.15)-- (8.33,17.68);
\draw [line width=1pt] (8.33,17.68)-- (8.33,16.75);
\draw [line width=1pt] (8.33,17.68)-- (9.14,18.14);
\draw [line width=1pt] (9.14,18.14)-- (9.95,17.68);
\draw [line width=1pt] (9.95,17.68)-- (9.94,16.75);
\draw [line width=1pt] (9.95,17.68)-- (10.75,18.14);
\draw [line width=1pt] (10.75,18.14)-- (11.56,17.68);
\draw [line width=1pt] (11.56,17.68)-- (11.56,16.74);
\draw [line width=1pt] (11.56,17.68)-- (12.36,18.14);
\draw [line width=1pt] (12.36,18.14)-- (13.17,17.67);
\draw [line width=1pt] (13.17,17.67)-- (13.97,18.14);
\draw [line width=1pt] (13.17,17.67)-- (13.17,16.74);
\draw [line width=1pt] (13.97,18.14)-- (13.98,19.07);
\draw [line width=1pt] (13.98,19.07)-- (13.17,19.53);
\draw [line width=1pt] (13.17,19.53)-- (12.36,19.07);
\draw [line width=1pt] (12.36,19.07)-- (12.36,18.14);
\draw [line width=1pt] (12.36,19.07)-- (11.56,19.54);
\draw [line width=1pt] (11.56,19.54)-- (10.75,19.07);
\draw [line width=1pt] (10.75,19.07)-- (10.75,18.14);
\draw [line width=1pt] (10.75,19.07)-- (9.95,19.54);
\draw [line width=1pt] (9.95,19.54)-- (9.14,19.07);
\draw [line width=1pt] (9.14,19.07)-- (9.14,18.14);
\draw [line width=1pt] (9.14,19.07)-- (8.34,19.54);
\draw [line width=1pt] (8.34,19.54)-- (7.53,19.08);
\draw [line width=1pt] (7.53,19.08)-- (7.53,18.15);
\draw [line width=1pt] (7.53,19.08)-- (6.73,19.54);
\draw [line width=1pt] (6.73,19.54)-- (5.92,19.08);
\draw [line width=1pt] (5.92,19.08)-- (5.92,18.15);
\draw [line width=1pt] (5.92,19.08)-- (5.11,19.54);
\draw [line width=1pt] (5.11,19.54)-- (5.12,20.47);
\draw [line width=1pt] (5.12,20.47)-- (5.92,20.94);
\draw [line width=1pt] (5.92,20.94)-- (6.73,20.47);
\draw [line width=1pt] (6.73,20.47)-- (6.73,19.54);
\draw [line width=1pt] (6.73,20.47)-- (7.53,20.94);
\draw [line width=1pt] (7.53,20.94)-- (8.34,20.47);
\draw [line width=1pt] (8.34,20.47)-- (8.34,19.54);
\draw [line width=1pt] (8.34,20.47)-- (9.14,20.93);
\draw [line width=1pt] (9.14,20.93)-- (9.95,20.47);
\draw [line width=1pt] (9.95,20.47)-- (9.95,19.54);
\draw [line width=1pt] (9.95,20.47)-- (10.76,20.93);
\draw [line width=1pt] (10.76,20.93)-- (11.56,20.47);
\draw [line width=1pt] (11.56,20.47)-- (11.56,19.54);
\draw [line width=1pt] (11.56,20.47)-- (12.37,20.93);
\draw [line width=1pt] (12.37,20.93)-- (13.17,20.46);
\draw [line width=1pt] (13.17,20.46)-- (13.17,19.53);
\draw [line width=2pt] (13.16,13.95)-- (13.17,20.46);
\draw [line width=2pt] (5.92,20.47)-- (13.17,20.46);
\draw [line width=2pt] (5.92,20.47)-- (5.91,13.96);
\draw [line width=2pt] (5.91,13.96)-- (13.16,13.95);
\draw [line width=1pt] (5.91,14.43) circle (1.18cm);
\draw [line width=1pt] (7.53,14.42) circle (1.18cm);
\draw [line width=1pt] (9.14,14.42) circle (1.18cm);
\draw [line width=1pt] (10.75,14.42) circle (1.18cm);
\draw [line width=1pt] (12.36,14.42) circle (1.18cm);
\draw [line width=1pt] (13.17,15.81) circle (1.18cm);
\draw [line width=1pt] (11.56,15.81) circle (1.18cm);
\draw [line width=1pt] (9.94,15.82) circle (1.18cm);
\draw [line width=1pt] (8.33,15.82) circle (1.18cm);
\draw [line width=1pt] (6.72,15.82) circle (1.18cm);
\draw [line width=1pt] (5.92,17.22) circle (1.18cm);
\draw [line width=1pt] (7.53,17.21) circle (1.18cm);
\draw [line width=1pt] (9.14,17.21) circle (1.18cm);
\draw [line width=1pt] (10.75,17.21) circle (1.18cm);
\draw [line width=1pt] (12.36,17.21) circle (1.18cm);
\draw [line width=1pt] (13.17,18.6) circle (1.18cm);
\draw [line width=1pt] (11.56,18.61) circle (1.18cm);
\draw [line width=1pt] (9.95,18.61) circle (1.18cm);
\draw [line width=1pt] (8.34,18.61) circle (1.18cm);
\draw [line width=1pt] (6.72,18.61) circle (1.18cm);
\draw [line width=1pt] (5.92,20.01) circle (1.18cm);
\draw [line width=1pt] (7.53,20.01) circle (1.18cm);
\draw [line width=1pt] (9.14,20) circle (1.18cm);
\draw [line width=1pt] (10.75,20) circle (1.18cm);
\draw [line width=1pt] (12.37,20) circle (1.18cm);
\end{tikzpicture} } \caption{Construction of a rectangle for $S_n>\frac{3 \sqrt 3} 2 n-\frac{17 \sqrt 3}{4} \sqrt n$.}\label{fig:1}
\end{figure}
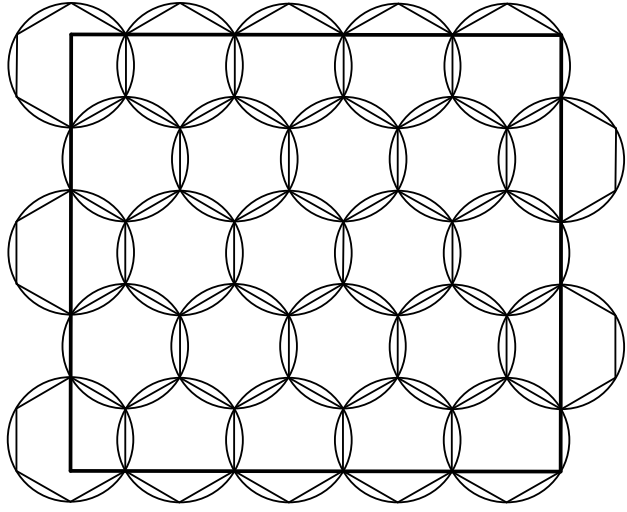

Proving the optimality is easy: take $\left \lfloor \sqrt n \right \rfloor ^2>n-2 \sqrt n$ disks and discard the others; with these, build a hexagonal lattice with $k=\left \lfloor \sqrt n \right \rfloor$ circles intersecting each side of the rectangle, placed as exemplified in Figure \ref{fig:1}. The resulting rectangle has area 
\[ S' _n= \left ( \frac 3 2 k- \frac 1 2 \right ) \left ( \sqrt 3 k-\frac {\sqrt 3} 2\right )>\frac {3 \sqrt 3} 2 k^2-\frac {5 \sqrt 3} 4 k.\] 
Since $n-2 \sqrt n<k^2 \leq n$, this is greater than $(3 \sqrt 3 /2 )n-(17 \sqrt 3 / 4) \sqrt n$; then $\lim_{n \rightarrow \infty}\allowbreak \frac {S' _n} {(3 \sqrt 3/2) n}=1$.
\end{proof}

As to Theorem \ref{2}, the proof of Theorem \ref{1} immediately shows $\alpha, \beta \geq 2 \sqrt 2 (K_6-K_5)$. Repeating the entire proof using $S_n \geq (3 \sqrt 3/ 2) n+O(\sqrt n)$ instead of $S_n \geq {2n}$ gives the sharper $\alpha, \beta \geq 2 \sqrt {K_6} (K_6-K_5)$. On the other hand, the construction we exhibited shows $\beta \leq 17 \sqrt 3 /4$. 

\begin{proof}[Proof of Theorem \ref{2}.]

First we deal with the upper bounds. Take the construction we performed before with $\left \lfloor c_1 \sqrt n \right \rfloor$ circles on a side and $\left \lfloor c_2 \sqrt n \right \rfloor$ on the other, with $c_1c_2=1$; using $\left \lfloor x \right \rfloor >x-1$ we get a constant of at most $(9 \sqrt 3/ 4) c_1+2 \sqrt 3 c_2$ for the $\sqrt n$ term. If $c_1=\sqrt{2/ 3}$ and $n/6$ is a perfect square, we lose nothing in taking the integer parts and the constant is $(\sqrt 3 /2) (3c_1 /2+c_2)= 3 /\sqrt 2 \geq \alpha$.

Next, it is easy to see that there is a perfect square between $n$ and $n-2 \sqrt n$: hence we can achieve that $n$ is a perfect square by discarding at most $2 \sqrt {n}$ disks; having done this, we can build, as we just did, an arrangement with an implicit constant for $\sqrt n$ of $5 \sqrt 3 /4$, so that $\beta \leq  2+5 \sqrt 3 /4$.

We now prove the lower bounds, strengthening those that can be obtained with the methods in \cite{V}: our method is almost the same as that of Verblunsky, but he only takes into account one side per cell instead of two (see below for the meaning of this), which results in weaker bounds. \\Label cyclically $\vor ^\flat _1$, $\dots$, $\vor ^\flat _\omega$ the Voronoi cells which have at least one side lying on the boundary of the rectangle, and $K(\vor_i ^\flat)$ the area of a cell $\vor_i ^\flat$. We shall hereafter suppose that every such cell has exactly one side lying on the perimeter of the triangle: this is the case for all cells but the four ones which contain the vertices of the rectangle; since their number is finite, they are irrelevant in our discussion and we will implicitly ignore them.

We rewrite the upper bound for the average number of sides as $6- \omega / \sqrt n + 2 /n$; again, the $ 2 /n$ summand will eventually give a constant contribution and we can omit it. Hence
\[\text{area of the rectangle} \leq \sum_{i=3}^\infty K_i V_i-\sum_{\flat} (K_{n_i}-K(\vor_i ^\flat))<K_6\sum_{i=3}^\infty V_i-\Sigma\]  
where the subscript $\flat$ with a sum indicates that it ranges over the boundary cells, $\Sigma=\sum_\flat (K_6-K_5+K_{n_i}-K(\vor_i ^\flat))$, and $n_i$ is the number of sides of $\vor_i ^\flat$.

We make another simple geometric observation, in the same spirit of the one we previously stated: if any $c$ consecutive sides of a cyclic $d$-agon are fixed, then it has the maximum area when the remaining $d-c$ sides all have equal length.

Let $\ell_i$ be the length of the side $\vor_i^\flat$ has on the boundary of the rectangle. Call $\tau(\ell_i, \ell_{i+1})$ the length of the side $\vor_i^\flat$ and $\vor_{i+1}^\flat$ have in common. 

Note next that the area of the triangle formed by a chord inside a circle (in any position) and the center of the circle is maximised when the endpoints of the chord lie on the boundary of the disk; to see this, first translate the chord along its line (thereby preserving the area) until the center of the circle belongs to its axis, then move it away from the center.

We will see later that there is a bound which does not depend on the number of sides of a cell, nor on the length of all but three consecutive sides of the cell (the side it shares with the boundary of the rectangle, and the two adjacent ones), so as long as we are only concerned about the length of three of the original sides, we can forget about the overall number or specific configuration of the sides, and insert or delete sides as well as change their configuration, if we do not change the length of the three  aforementioned sides. Because of this (and convexity reasons), for each boundary cell there is a cell where the three sides we care about are of the same length but with all of its vertices on the boundary of a circle--and the lengths and number of the other sides may have changed--with its area not smaller than the area of the original one. Along the same lines it can be seen geometrically that one can work as if the endpoints of the chords on the boundary of the rectangle were intersections of their circumscribing circles (see \cite{V} for the details). We may thus consider such a new cell as $\vor_i^\flat$ instead, since we only keep track of the data we mentioned above, and in this case 
\[  K(\vor_i^\flat) \leq K(\ell_i)+K(\tau(\ell_{i-1}, \ell_i))+K(\text{all the other $n_i-2$ sides equal}) \]
where $K(\ell)$ is the area of the triangle formed by the center of the circle and a chord of length $\ell$, or the sum of the areas of the remaining triangles as in the last summand (the two uses being clear from the context).
The central angle of a chord of length $\ell$ is $\vartheta(\ell)=\arccos \left ( 1- {\ell^2} /2 \right )$. Now, we have the following (by elementary Euclidean geometry):
\begin{gather*} K_{n_i}=\frac {n_i} 2 \sin \frac {2 \uppi}{n_i}   ;
\\K(\ell_i)=\frac {\ell_i} 4 \sqrt{4-\ell_i ^2};
\\ \tau(\ell_{i-1}, \ell_i)= \sqrt {\frac {4-\ell_i \ell_{i-1}+\sqrt{(4-\ell_i ^2)(4-\ell_{i-1}^2)}} 2};
\\  K(\tau(\ell_{i-1}, \ell_i))=\frac {{\tau(\ell_{i-1}, \ell_i)}} 4 \sqrt{4-{\tau(\ell_{i-1}, \ell_i)} ^2} \\ =\frac 1 {4 \sqrt 2} \sqrt{2(\ell_i ^2+\ell_{i-1} ^2)-\ell_i ^2 \ell_{i-1} ^2 +\ell_i \ell_{i-1} \sqrt{(4-\ell_i ^2)(4-\ell_{i-1} ^2)}} ; 
\\ K(\text{all the other $n_i-2$ sides equal})=\frac {n_i-2}{2} \sin \frac {2 \uppi-\vartheta(\tau(\ell_{i-1}, \ell_i))-\vartheta(\ell_i)}{n_i-2} 
\\=\frac {n_i-2}{2} \sin \frac {2 \uppi-\arccos\frac{\ell_{i-1}\ell_i-\sqrt{(4-\ell_{i-1} ^2)(4-\ell_i ^2)}}{4}-\arccos(1-\frac{\ell_i ^2}2)}{n_i-2} .
\end{gather*}
At this moment we need to minimize the cyclic sum
\begin{gather*}
\sum_\flat \left( (K_6-K_5)+K_{n_i}- K(\ell_i)-K(\tau(\ell_{i-1}, \ell_i))\right. \\ \left. -K(\text{all the other $n_i-2$ sides equal}) \right )
\\=\sum_\flat \left ( (K_6-K_5)+\frac {n_i} 2 \sin \frac {2 \uppi}{n_i} -\frac {\ell_i} 4 \sqrt{4-\ell_i ^2} \right.
\\-\frac 1 {4 \sqrt 2} \sqrt{2(\ell_i ^2+\ell_{i-1} ^2)-\ell_i ^2 \ell_{i-1} ^2 +\ell_i \ell_{i-1} \sqrt{(4-\ell_i ^2)(4-\ell_{i-1} ^2)}}\\\left. -\frac {n_i-2}{2} \sin \frac {2 \uppi-\arccos\frac{\ell_{i-1}\ell_i-\sqrt{(4-\ell_{i-1} ^2)(4-\ell_i ^2)}}{4}-\arccos(1-\frac{\ell_i ^2}2)}{n_i-2} \right )  .
\end{gather*}

We first get rid of the $n_i$'s. Call $\theta=\theta_i:=\arccos \frac{\ell_{i-1}\ell_i-\sqrt{(4-\ell_{i-1} ^2)(4-\ell_i ^2)}}{4}+\arccos ( 1-\ell_i ^2/ 2 )$, and define \[\rho_\theta (x) :=x \sin \frac {2 \uppi} x-(x-2) \sin \frac{2 \uppi-\theta}{x-2}.\]
Then $\frac {\text{d} \rho_\theta} {\text{d} x} =\eta \left ( \frac {2 \uppi} x\right )-\eta \left (\frac {2 \uppi-\theta}{x-2}\right )$, where $\eta(y)=\sin y-y \cos y$. If $x=3$ then $\rho _\theta(3)=3 \sqrt 3/ 2+\sin \theta$: since this is always greater than $2\sin ( \theta/ 2)$, the lower bound we are going to find for the other values of $x$, this case is settled. If $x \geq 4$, the arguments of $\eta$ are between $0$ and $\uppi$ and $\eta$ is monotone increasing: having the derivative equal to $0$ means $ {2 \uppi} /x=(2 \uppi-\theta)/ (x-2)$, or $\theta x=4 \uppi$. Indeed, for $n_i=4 \uppi/\theta_i$, $\rho_{\theta_i}(n_i)$ has a minimum, and this is the way we eliminate the $n_i$'s; what is left is the lower bound
\begin{gather*}
\Sigma'=\sum_\flat \left ( (K_6-K_5) -\frac {\ell_i} 4 \sqrt{4-\ell_i ^2} \right. \\ -\frac 1 {4 \sqrt 2} \sqrt{2(\ell_i ^2+\ell_{i-1} ^2)-\ell_i ^2 \ell_{i-1} ^2 +\ell_i \ell_{i-1} \sqrt{(4-\ell_i ^2)(4-\ell_{i-1} ^2)}}\\ \left. + \sin \frac {\arccos\frac{\ell_{i-1}\ell_i-\sqrt{(4-\ell_{i-1} ^2)(4-\ell_i ^2)}}{4}+\arccos(1-\frac{\ell_i ^2}2)}{2} \right ) .
\end{gather*}

We want to minimize the cyclic sum $\sum_i f(\ell_{i}, \ell_{i-1})$, where
\begin{gather*} f(x,y) = -\frac {x} 4 \sqrt{4-x ^2}-\frac 1 {4 \sqrt 2} \sqrt{2(x ^2+y ^2)-x ^2 y ^2 +x y \sqrt{(4-x ^2)(4-y ^2)}}\\+ \sin \frac {\arccos\frac{xy-\sqrt{(4-x ^2)(4-y ^2)}}{4}+\arccos(1-\frac{x ^2}2)}{2}; \end{gather*}
we treat the perimeter $P$ of the rectangle as fixed. It would be useful to prove that the minimum of $\Sigma'$ can be attained when all the variables are equal; in order to do this, shif half of each $-( x/ 4) \sqrt{4-x^2}$ term to the next summand, so that $\sum_i f(\ell_{i}, \ell_{i-1})=\sum_i \tilde{f}(\ell_{i}, \ell_{i-1})$ for 
\begin{gather*} \tilde{f}(x,y) = -\frac {x} 8 \sqrt{4-x ^2}-\frac {y} 8 \sqrt{4-y ^2}\\-\frac 1 {4 \sqrt 2} \sqrt{2(x ^2+y ^2)-x ^2 y ^2 +x y \sqrt{(4-x ^2)(4-y ^2)}}\\+ \sin \frac {\arccos\frac{xy-\sqrt{(4-x ^2)(4-y ^2)}}{4}+\arccos(1-\frac{x ^2}2)}{2}. \end{gather*}

\begin{figure}[t]
\scalebox{0.8}{

\includegraphics{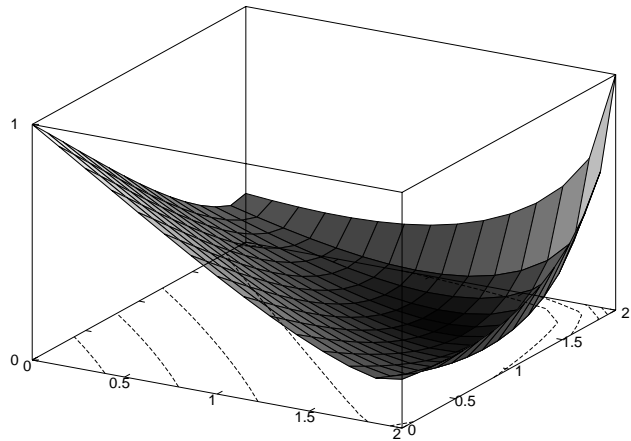}

}

\caption{Plot and level sets of $\tilde{f}$.} \label{fig:4}
\end{figure}

This is not symmetric in $x$ and $y$, but it differs everywhere for less than $10^{-7}$ from a symmetric function; this fact will be verified by explicit calculation to affect the final constant just beyond the seventh decimal digit--since if each summand is changed by at most $10^{-7}$ so does their mean. Keeping this in mind, we may work as if $\tilde{f}$ was symmetric.

The function $\tilde{f}$ may not be \textit{a priori} convex everywhere when $x+y<2$ (see Figure \ref{fig:4} for the plot): since direct calculation on $\tilde{f}$ is a significant effort, to deal with this, define instead another function 
\[ g(x,y) := 
\begin{cases}
\tilde{f}(x,y)   \mbox{ if } x+y \geq 2 , \\
\mbox{the largest (weakly) convex function such that $g \leq \tilde{f}$, if } x+y<2.

\end{cases}
\]

After all these tweaks, we can finally say that for fixed $x+y$, $g$ has a minimum for $x=y$: it follows that we may restrict ourselves \textit{a posteriori} to the line $\ell_1=\ell_2= \dots =\ell _\omega$, as long as $g$ has its minimum where $\tilde{f} \equiv g$ (as we will see).

Now note that for $ \ell=   P/ \omega$, $\ell_{\text{min}}$ the value of $\ell$ for which our minimum is attained, and $\lambda :=K_6-K_5+1$ we have
\[ \Sigma ' \geq \sum _\flat \left ( (K_6-K_5)+g(\ell_{\text{min}}, \ell_{\text{min}}) \right ) =\omega  \lambda-\frac P 2 \sqrt {4- \left ( \frac P \omega \right ) ^2} .\]
This holds again because the minimum we are looking for is in the region where $g \equiv \tilde{f}$, which can be seen by noting that in $x+y<2$ the partial derivatives of $f$ and $g$ are negative (by termwise differentiation).
\\Moreover, $\frac {\text{d}} {\text{d} \omega} \left ( \omega  \lambda-\frac P 2 \sqrt {4- \left ( \frac P \omega \right ) ^2}   \right )=\lambda-\frac {P^3}{2 \omega ^3 \sqrt {4-\left ( \frac P \omega  \right ) ^2}}$: the minimum is found for $P/ \omega =\ell_{\text{min}}=1.484490 \ldots=:u_0$, which is a root of the equation $u^6+4\lambda^2 u^2-16\lambda^2=0$ in $u:= P /\omega$. In this point, the sum equals $\omega \left (K_6-K_5+1- \frac {u_0} 2 \sqrt {4-u_0 ^2} \right )= \omega \cdot 0.225635 \dots$ and the statement of the theorem follows upon multiplying by $2 \sqrt {K_6}$.
\end{proof}

Finally, note that if one considers only rectangles whose side ratio tends to a constant $\psi$ as $n \rightarrow \infty$, the lower bounds can be easily improved by a factor $(\sqrt{\psi}+1/\sqrt{\psi})/2$, by modifying the bound in the inequality for the average number of sides.

The problem of finding the exact value of $S_n$ for small $n$ is solved only for $n \leq 5$ (and in those cases one has equality in Theorem \ref{1}), see \cite{H} and \cite{M}. The same methods apply--and give similar results--when instead the rectangle is fixed and we seek for the least number of unit disks which can cover it. We do not expect a significant improvement of our bounds to be possible without employing entirely new ideas.


\begin{thebibliography}{10}

\bibitem {G} M. Goldberg, \textit{The isoperimetric problem for polyhedra}, Tohoku Math. J. \textbf{40} (1934), 228-229.

\bibitem {H} A. Heppes and H. Melissen, \textit{Covering a rectangle with equal circles}, Period. Math. Hung. \textbf{34.1} (1997), 65-81.

\bibitem {K} R. Kershner, \textit{The number of circles covering a set}, Amer. J. Math. \textbf{61.3} (1939), 665-671.

\bibitem {M} J. B. M. Melissen and P. C. Schuur, \textit{Covering a rectangle with six and seven circles}, Discrete Appl. Math. \textbf{99.1} (2000), 149-156.


\bibitem{V} S. Verblunsky, \textit{On the least number of unit circles which can cover a square}, J. London Math. Soc. \textbf{24.3} (1949), 164-170.






\end{thebibliography}
\end{document}